\documentclass{amsart}

\usepackage{enumitem}
\usepackage[cmtip,all]{xy}
\usepackage{amsmath, amsthm, latexsym, amssymb, amsfonts}

\usepackage{alltt}
\usepackage[pdftex]{graphicx,color} 
\usepackage{caption}
\usepackage{setspace}
\usepackage{color,hyperref}
\usepackage{longtable}
\usepackage[table]{xcolor}
\usepackage{float}

\newcommand{\reg}{\operatorname{reg}}

\newcommand{\rank}{\operatorname{rank}}

  \newcommand{\lcm}{\operatorname{lcm}}

\newcommand{\Z}{\mathbb Z}
\newcommand{\N}{\mathbb N}

\newcommand{\F}{\mathbb F}

\newcommand{\K}{{\mathbb K}}

\def\ev{{\text{ev}}}

\newtheorem{theorem}{Theorem}[section]
\newtheorem{lemma}[theorem]{Lemma}
\newtheorem{corollary}[theorem]{Corollary}

\newtheorem{proposition}[theorem]{Proposition}
\newtheorem{definition}[theorem]{Definition}

\newtheorem{example}[theorem]{Example}
\newtheorem{remark}[theorem]{Remark}

\numberwithin{equation}{section}

\newenvironment{proofA}{\noindent\textit{Proof of Theorem \ref{1abfree}}}{\hfill{$\Box$}}
\newenvironment{proofC}{\noindent\textit{Proof of Proposition \ref{hilbertseries}}}{\hfill{$\Box$}}

\renewcommand{\P}{{\mathbb P}}
\usepackage{cleveref}
\colorlet{darkred}{red!80!black}
\colorlet{darkblue}{blue!80!black}
\colorlet{darkgreen}{green!40!black}

\begin{document}

\title{Algebraic Invariants of Codes on Weighted
Projective Planes}


\author{Ya\u{g}mur \c{C}ak\i ro\u{g}lu}
\address{ Department of Mathematics,
  Hacettepe  University,
  Ankara, TURKEY}
\curraddr{}
\email{yagmur.cakiroglu@hacettepe.edu.tr}

\author{Mesut \c{S}ah\.{i}n}
\address{ Department of Mathematics,
  Hacettepe  University,
  Ankara, TURKEY}
\curraddr{}
\email{mesut.sahin@hacettepe.edu.tr}
\thanks{The authors are supported by T\"{U}B\.{I}TAK Project No:119F177}

\date{}

\dedicatory{}

\begin{abstract} Weighted projective spaces are natural generalizations of projective spaces with a rich structure. Projective Reed-Muller codes are error-correcting codes that played an important role in reliably transmitting information on digital communication channels. In this case study, we explore the power of commutative and homological algebraic techniques to study weighted projective Reed-Muller (WPRM) codes on weighted projective spaces of the form $\P(1,1,b) $. We compute minimal free resolutions and thereby obtain Hilbert series for the vanishing ideal of the $\F_q$-rational points, and compute main parameters for these codes.
 
\end{abstract}
\maketitle 

\section{Introduction}
Let $q$ be a prime power and $\F_q$ denote the finite field with $q$ elements. Denote by $\K$ the algebraic closure of $\F_q.$ Let $w_1,\dots,w_r$ be positive integers with a trivial greatest common divisor. The \textbf{weighted projective space} associated to the weights $w_1,\dots,w_r$ is the quotient space \[\P(w_1,\dots,w_r)=(\K^{r}\setminus \{0\})/\K^*\] under the equivalence relation \begin{equation*}
(x_1,\dots,x_r)\sim(\lambda^{w_1}x_1,\dots,\lambda^{w_r}x_r) \mbox{ for } \lambda \in \K^*.
\end{equation*}
So, points in the space $X=\P(w_1,\dots,w_r)$ are equivalence classes $[x_1:\dots:x_r]$ of points $(x_1,\dots,x_r)\in \K^{r}\setminus \{0\}.$ It is known that elements $[x_1:\dots:x_r]$ of the $\F_q$-rational points $X(\F_q)$ have representatives with coordinates $x_1,\dots,x_r \in \F_q$, see \cite{ACGLOR2017,BRWPS}.

We work with the polynomial ring $S=\F_q[x_1,\dots,x_r]$, which is graded via $\deg x_i =w_i$ for all $i=1,\dots,r$. Thus, we have the following decomposition
\begin{equation*}\label{homgrdring}
S= \bigoplus_{d\in \N W} S_d
\end{equation*}
where $S_d$ is the vector space whose basis consists of the monomials $x_1^{m_1}\cdots x_r^{m_r}$ of degree $d=m_1w_1+\cdots+m_r w_r$ and $\N W:=\langle w_1,\dots,w_r \rangle$ is the subsemigroup of $\N$ generated by the weights. 

Take any subset $Y=\{P_1,\dots,P_N\}\subseteq X(\F_q)$ of $\F_q$-rational points of $X$ and define the evaluation map as follows

\begin{align}\label{evmap}
ev_Y&:\begin{cases}
S_d\rightarrow \F_q^{N}\\
f\mapsto\left(f(P_1),\dots,f(P_N)\right)  
\end{cases}
\end{align}
The image $C_{d,Y}$ is a subspace of $\F_q^{N}$ known as the \textbf{evaluation code} obtained by evaluating homogeneous polynomials of degree $d$ at the points of $Y$. There are three main parameters $[N,K,\delta]$ of a linear code. The \textbf{length} $N$ is defined by the cardinality $|Y|$ of the subset in our case.
The \textbf{dimension} of $C_{d,Y}$ (a measure of efficiency), denoted $K=\dim_{\F_q}C_{d,Y}$ is the dimension as a subspace of $\F_{q}^{N}$.
The number of non-zero entries in any $c\in C_{d,Y}$ is called its \textbf{Hamming weight} and the
\textbf{minimum distance} $\delta$ of $C_{d,Y}$ (a measure of reliability) is the smallest weight among all code words $c\in C_{d,Y}\backslash\{0\}$.
There is an algebraic approach for studying these codes relying on the vanishing ideal $I(Y)$ of $Y$ which is the graded ideal generated by homogeneous polynomials in $S$ vanishing at every point of $Y$. Since the kernel of the evaluation map defined in \eqref{evmap} equals the homogeneous piece $I_d(Y)$ of degree $d$, we have an isomorphism $S_d/I_d(Y)\cong C_{d,Y} $. Thus, the dimension of $C_{d,Y}$ is the weighted Hilbert function defined in Definition \ref{def:hil} 
\begin{equation}\label{dimhil}
\dim_{\F_q}C_{d,Y}=H_{Y}(d).
\end{equation}
The length $N$ of $C_{d,Y}$ is given to be\[N=q^{r-1}+q^{r-2}+\cdots+q+1\] when $Y=X(\F_q)$, see \cite[Section 2.1]{ACGLOR2017}.

Projective Reed-Muller codes on a finite field which are extensions of the classical generalized Reed-Muller codes are introduced by Lachaud in the paper \cite{GL1988}. These codes are obtained by evaluating homogeneous polynomials at all the $\F_q$-rational points of a given projective space. Lachaud \cite{GL1990} has given the parameters of these codes. Reed-Muller codes are error-correcting codes that played an important role in transmitting information on digital communication channels reliably. Projective Reed-Muller codes are well-studied codes (see, \cite{S91,S1992,GL1990,GL1988,A92,BSS2016}), used in real life applications.

Weighted projective spaces are natural generalizations of classical projective spaces, having rich structures and exhibiting interesting algebraic geometric properties. In literature, weighted projective spaces are considered convenient ambient spaces to create interesting classes of linear codes over finite fields and the Weighted Projective Reed-Muller codes $C_{d,Y}$ as defined above has been studied for the first time in \cite{ACGLOR2017}. Parameters of the code on the weighted projective plane $\P(1,a,b)$ of degree $d=k \cdot \lcm(a,b) \leq q$ for some positive integer $k$, were presented in the paper \cite{ACGLOR2017}. Although there are codes bearing the same name introduced by S\o rensen in \cite{S1992}, they are indeed different, see also \cite{GT2013}. A related code has been studied by Dias and Neves \cite{DN2017} who focus on the case where $Y$ is the set of $\F_q$-rational points of the projective weighted torus. They proved that the vanishing ideal is a special binomial ideal, which is generalized later by the second author to the vanishing ideal of the $\F_q$-rational points of the torus of a more general toric variety, see \cite{MS2018}. Nardi offered to extend the length of a toric code parallel to the extension of Reed-Muller codes into projective ones, by evaluating at the full set of $\F_q$-rational points in the paper \cite{JN22}. The parameters of the codes from Hirzebruch surfaces when the evaluation set is the full set of $\F_q$-rational points were computed by Nardi in a previous paper \cite{JN2019}.

In this paper, our main goal is to examine algebraic invariants of the vanishing ideal of $\F_q$-rational points of the weighted projective plane $X=\P(1,a,b)$ and their relation with the parameters of the code whose degree does not have to satisfy the condition $d=k \cdot \lcm(a,b) \leq q$ as in \cite{ACGLOR2017}. Algebraic invariants we focus on are minimal free resolutions, Betti numbers, Hilbert series and regularity.  We pay particular attention to the
case $\P(1,1,b)$ to obtain more explicit information. This yields the Hilbert function giving the dimension of the code and regularity set which is crucial to eliminate trivial codes. 

The paper is organized as follows. We give definitions of the algebraic invariants in Section \ref{sec:pre}. We present the main results for algebraic invariants of the vanishing ideals of the set of $\F_q$-rational points $Y=\P(1,1,b)(\F_q)$ in Section \ref{sec:alginv}. In particular, we give the regularity set of $\P(1,1,b) (\F_q)$ as a consequence. We also compute the dimensions of the codes over $Y$ for any degree $d$. 
In Section \ref{sec:minimum} we compute the minimum distance of these codes and share  tables of their main paramaters for $b=2,5,7$ and $q=2,5$.

\section{Preliminaries}\label{sec:pre}
In this section, we present some definitions which is needed to state our main results. Let $X=\P(w_1,\dots,w_r)$ be a weighted projective space and $Y=X(\F_q)$ be the set of $\F_q$-rational points.  The first algebraic invariant of $Y$ which has a key connection with the homogeneous polynomials corresponding to zero code-words is the following ideal. 

\begin{definition} The vanishing ideal $I(Y)$ is the (homogeneous) ideal of the ring $S=\F_q[x_1,\dots,x_r]$ generated by homogeneous polynomials vanishing on $Y$, so 
\[I(Y)=\bigoplus_{d\in \N W} I_d(Y),\]
where $I_d(Y)$ is the homogeneous piece of degree $d$ and $\N W$ is the subsemigroup of $\N$ generated by the weights.
\end{definition}
The second algebraic invariant of $Y$ related to the parameters of the code is given next.
\begin{definition}\label{def:hil} Let $S_d$ be the vector space spanned by monomials $x_1^{m_1}\cdots x_r^{m_r}$ of degree $d=m_1w_1+\cdots+m_r w_r$ and $I_d(Y)$ be the vector space of homogeneous polynomials of degree $d$ vanishing on $Y$. The (weighted) Hilbert function of $Y$ is defined as follows \[H_{Y}(d)=\dim_{\F_q}{S}_{d}-\dim_{\F_q} I_{d}(Y).\] Also, the (weighted) Hilbert series of $Y$ is defined by 
$$HS_{Y}(d)=\sum\limits_{d\in \N W}H_{Y}(d)t^{d}.$$
\end{definition}
The last algebraic invariant of $Y$ used to eliminate trivial codes is the following set.
\begin{definition}
The regularity set of $Y$ is defined as follows \[\reg (Y)=\{ d\in \N W : H_{Y}(d)=|Y|\}.\]
\end{definition}

We use the following important criterion to prove exactness of the minimal free resolution of the vanishing ideal.
\begin{lemma}
    \label{BE1973}\cite[Corollary 2]{BE1973} Let 
\begin{equation*}
0\rightarrow F_n\xrightarrow{\phi_n}
F_{n-1}\xrightarrow{\phi_{n-1}
}\cdots\xrightarrow{\phi_2} F_1\xrightarrow{\phi_1}F_0
\end{equation*}
be a complex of free modules over a Noetherian ring $S$. Let $\rank (\phi_i)$ be the size of the largest nonzero minor in the matrix describing $\phi_i$, and let $I(\phi_i)$ be the ideal generated by the minors of maximal rank. Then the complex is exact if and only if for all $1\le i\le n$
\begin{enumerate}
    \item $\rank(\phi_{i+1})+\rank(\phi_{i})=\rank(F_i)$
    \item $I(\phi_i)$ contains an $S$-sequence of length $i$.
\end{enumerate}
\end{lemma}

When the cokernel $M$ of $\phi_1$ in the exact complex 
\begin{equation*}
0\rightarrow F_n\xrightarrow{\phi_n}
F_{n-1}\xrightarrow{\phi_{n-1}
}\cdots\xrightarrow{\phi_2} F_1\xrightarrow{\phi_1}F_0
\end{equation*}
 is a graded module, the free modules are of the form $$F_i=\bigoplus_{j=1}^{\rank F_i} S(-d_{ij}).$$ The number of indices $j$ for which $d_{ij}=d\in \N W$ is called the graded $i$-Betti number of $M$ and is denoted by $\beta_{i,d}(M)$. All but finitely many elements $d\in \N W$ correspond to the trivial graded Betti number $\beta_{i,d}(M)=0$. 
 
 We use the following result to give the Hilbert series of $Y$ as a rational function using the minimal free resolution of $I(Y)$.

\begin{theorem}\cite[Theorem 8.20, Proposition 8.23]{MS2005}\label{MS} The Hilbert series of a finitely generated graded module $M$ over a polynomial ring $S=\F_q[x_1,\dots,x_r]$ positively graded by $\N W$ is a rational function given by,
\[HS_{M}(t)=\frac{\mathcal{K}_M(t)}{\prod\limits_{i=1}^{r}(1-t^{w_i})},\]
where the polynomial $\mathcal{K}_M(t)=\sum\limits_{i=0}^n\sum\limits_{d\in \N W}(-1)^i\beta_{i,d}(M)t^{d}$.
\end{theorem}

\section{Algebraic Invariants relevant to Weighted Projective Reed-Muller codes}\label{sec:alginv}
In this section, we compute algebraic invariants of the vanishing ideal. More precisely, we obtain the minimal free resolution as well as the Hilbert series of this ideal and use them to get the values of the Hilbert function. As a consequence, the regularity set of $Y=\P(1,1,b)(\F_q)$ is given which plays a crucial role in eliminating trivial (full) codes. We conclude the section by presenting the dimensions of the codes. 

\begin{theorem}\cite[Corollary 5.8]{MS2022}\label{1abideal} Let $a,b$ be positive integers with $\gcd{(a,b)}=1$, and $X=\P(1,a,b)$ be weighted projective plane over a finite field. Then, the vanishing ideal $I(X(\F_q))$ of $X(\F_q)$ is generated minimally by
\begin{equation*}
f_1=x_2^{(q-1)b+1}x_3-x_2x_3^{(q-1)a+1}, f_2=x_1^{(q-1)b+1}x_3-x_1x_3^{q} \text{ and } f_3=x_1^{(q-1)a+1}x_2-x_1x_2^{q}.
\end{equation*} 
\end{theorem}
\begin{corollary}\label{11bideal}
Let $a$ be a positive integer and $X=\P(1,1,b)$ be weighted projective space over a finite field. Then, the ideal of the set $X(\F_q)$ is $I(X(\F_q))=\langle f_1,f_2,f_3\rangle$,
where \begin{equation*}
 f_1=x_2^{(q-1)b+1}x_3-x_2x_3^{q}, f_2=x_1^{(q-1)b+1}x_3-x_1x_3^{q} \text{ and } f_3=x_1^qx_2-x_1x_2^{q}.    
\end{equation*}\end{corollary}
\begin{theorem}\label{1abfree} Denote by $I$ the ideal $I(X(\F_q))$ given by Theorem \ref{1abideal} for the weighted projective plane $X=\P(1,a,b)$. Then, the following complex is a graded minimal free resolution of $S/I$
\begin{equation*}
0\rightarrow \bigoplus_{j=1}^{2} S(-\sigma_j) \xrightarrow{\begin{bmatrix}
x_1& 0\\
A_1& f_3/x_1\\
A_2& -f_2/x_1
\end{bmatrix}} \bigoplus_{j=1}^{3} S(-\lambda_j)\xrightarrow{\begin{bmatrix}f_1 &f_2 &f_3\end{bmatrix}
}
S\rightarrow S/I\rightarrow0
\end{equation*} where $\lambda_{1}=(q-1)ab+a+b,\quad  \lambda_{2}=qb+1, \quad \lambda_{3}=qa+1$, 
\[
\begin{array}{llll}
  &A_1=-\sum\limits_{i=1}^{a}x_1^{(i-1)(q-1)b}x_2x_3^{(q-1)a-iq+i},\quad   &A_2=\sum\limits_{i=1}^{b}x_1^{(i-1)(q-1)a}x_2^{(q-1)b-iq+i}x_3,\\
  &\sigma_{1}=(q-1)ab+a+b+1, \quad &\sigma_{2}=qb+qa+1.
  \end{array}
\]
\end{theorem}
\begin{corollary}\label{11afree}
Let $I$ be the ideal of $X(\F_q)$ given by Corollary \ref{11bideal} for the weighted projective space $X=\P(1,1,b)$. Then, the following complex is a graded minimal free resolution of $S/I$ 
\begin{equation*}
0\rightarrow \bigoplus_{j=1}^{2} S(-\sigma_j) \xrightarrow{\begin{bmatrix}
x_1& 0\\
-x_2x_3^{q-1}& f_3/x_1\\
A& -f_2/x_1
\end{bmatrix}}\bigoplus_{j=1}^{3} S(-\lambda_j)\xrightarrow{\begin{bmatrix}f_1 &f_2 &f_3\end{bmatrix}
}S\rightarrow S/I\rightarrow0
\end{equation*} where 
\begin{align*}
A&=\sum\limits_{i=1}^{b}x_1^{(i-1)(q-1)}x_2^{(q-1)b-iq+i}x_3,\quad
&\sigma_{1}=qb+2, \quad \quad&\sigma_{2}=qb+q+1,\\
\lambda_{1}&=qb+1,\quad &\lambda_{2}=qb+1, \quad\quad
&\lambda_{3}=q+1.
\end{align*}

\end{corollary}

\begin{proofA} Let $\phi_1:=\begin{bmatrix}f_1 &f_2 &f_3\end{bmatrix}$ and $\phi_{2}:={\begin{bmatrix}
x_1& 0\\
A_1& f_3/x_1\\
A_2& -f_2/x_1
\end{bmatrix}}$ where 
\[f_1=x_2^{(q-1)b+1}x_3-x_2x_3^{(q-1)a+1}, f_2=x_1^{(q-1)b+1}x_3-x_1x_3^{q}, f_3=x_1^{(q-1)a+1}x_2-x_1x_2^{q} \text{ and} 
\]
\[A_1=-\sum\limits_{i=1}^{a}x_1^{(i-1)(q-1)b}x_2x_3^{(q-1)a-iq+i}, \quad A_2=\sum\limits_{i=1}^{b}x_1^{(i-1)(q-1)a}x_2^{(q-1)b-iq+i}x_3.
\]
We will use Lemma \ref{BE1973} to show the exactness of the following complex of graded modules.
 \begin{equation}\label{grRes}
0\rightarrow F_2 \xrightarrow{\phi_2} F_1\xrightarrow{\phi_1}S\rightarrow S/I\rightarrow 0,
\end{equation} where $F_2=S(-\sigma_1)\oplus S(-\sigma_2)$, $F_1=S(-\lambda_1)\oplus S(-\lambda_2)\oplus S(-\lambda_3)$.  
It is easily seen that rank$(\phi_2)=2$, rank$(\phi_1)=1$ and rank$(S^3)=3$. So we get $$\mbox{rank}(S^3)=\mbox{rank}(\phi_2)+\mbox{rank}(\phi_1).$$ Now we need to verify the second item in Lemma \ref{BE1973}. $I(\phi_1)$ is clearly $I=\langle f_1,f_2,f_3\rangle$ which contains a polynomial that is not a zero-divisor on $S$.  The minors of $\phi_2$ are
$$\begin{vmatrix} A_1& f_3/x_1\\
A_2& -f_2/x_1 
\end{vmatrix}=-f_1\neq 0,  \begin{vmatrix} x_1 & 0\\ 
A_2& -f_2/x_1 
\end{vmatrix}=-f_2\neq 0 \text{ and } 
\begin{vmatrix}
x_1 & 0\\ 
A_1 & f_3/x_1
\end{vmatrix}=f_3\neq 0.$$ Thus, $I(\phi_2)$ is again the ideal  $I=\langle f_1,f_2,f_3\rangle$ which contains the regular sequence $\{f_1+f_2,f_2+f_3\}$. This is true because if $f_2+f_3$ was a zero divisor in $S/\langle f_1+f_2 \rangle$, then there would be $g\in S\setminus \langle f_1+f_2 \rangle$ such that $(f_2+f_3)g\in \langle f_1+f_2 \rangle$, that is there would be $h\in S$ with $h(f_1+f_2)=g(f_2+f_3)$, which is clearly not possible. Because if there would be $h\in S$ satisfying this condition, these two polynomials $f_1+f_2$ and $f_2+f_3$ would have a common factor. But these two polynomials do not have a common factor. So the free resolution given in Theorem \ref{1abfree} is an exact sequence by Lemma \ref{BE1973}.

Now let us prove that the resolution is graded, that is, the maps $\phi_1$ and $\phi_2$ preserve degrees. We say that $\phi:M=\bigoplus M_d\rightarrow N=\bigoplus N_d$ is graded (or degree preserving) if $\phi(M_d)\subseteq N_d,$ for all $d\in \N W$, (see, \cite[Definition 8.12, Page 153]{MS2005}). So, we need to show that $\phi_1({(F_1)}_d)\subseteq {(F_0)}_d$. Recall that
\begin{eqnarray*}
 \phi_1:& \displaystyle F_1=\bigoplus_{j=1}^{3} S(-\lambda_j) \xrightarrow{\begin{bmatrix}f_1&f_2&f_3\end{bmatrix}}F_0=S\\
&(m_1,m_2,m_3)\rightarrow m_1f_1+m_2f_2+m_3f_3.
\end{eqnarray*}
If $\displaystyle (m_1,m_2,m_3)\in (F_1)_d= S_{d-\lambda_1}\oplus S_{d-\lambda_2} \oplus S_{d-\lambda_3}$ for $d\in \N W$, then we have
\[\deg(m_jf_j)=(d-\lambda_j)+\deg(f_j)=d, \text{ since } \lambda_j=\deg(f_j) \text{ for all } j=1,2,3.
\] 
Hence, $\phi_1$ is graded. Consider now the following map.
\begin{eqnarray*}
\phi_2:&\displaystyle F_2=\bigoplus_{j=1}^{2} S(-\sigma_j) \xrightarrow{\begin{bmatrix}
x_1& 0\\
A_1& f_3/x_1\\
A_2& -f_2/x_1
\end{bmatrix}}F_1=\bigoplus_{j=1}^{3} S(-\lambda_j) \\
&(m_1,m_2)\rightarrow(x_1m_1+0m_2, A_1m_1+m_2f_3/x_1, A_2m_1-m_2f_2/x_1)
\end{eqnarray*}
So as to show that $\phi_2({(F_2)}_d)\subseteq {(F_1)}_d$, we take $(m_1,m_2)\in {(F_2)}_d=S_{d-\sigma_1}\oplus S_{d-\sigma_2}$.

By the definition of the map $\phi_2$, we have 
\begin{eqnarray*}
 \deg(x_1m_1)=\deg(x_1)+\deg(m_1)=1+d-((q-1)ab+q+b+1)=d-\lambda_1.   
\end{eqnarray*}
Also,
\[\deg(A_1m_1)=\deg(m_2f_3/x_1)=qa+d-(qb+qa+1)=d-(qb+1)=d-\lambda_2.\]
And similarly, we have
\[\deg(A_2m_2)=\deg(-m_2f_2/x_1)=qb+d-(qa+qb+1)=d-(qa+1)=d-\lambda_3.\]

Therefore we get, \[ \phi_2({(F_2)}_d)=\phi_2(S_{d-\sigma_1}\oplus S_{d-\sigma_2})\subseteq S_{d-\lambda_1}\oplus S_{d-\lambda_2}\oplus S_{d-\lambda_3}=(F_1)_d\]
This proves $\phi_2$ is graded completing the proof.
\end{proofA}

 \begin{proposition}\label{hilbertseries}
 The formula of the Hilbert series of $Y=\P(1,a,b)(\F_q)$ is given as follows.
\[
HS_{Y}(t)=\frac{1-t^{qa+1}-t^{qb+1}-t^{(q-1)ab+a+b}+t^{qa+qb+1}+t^{(q-1)ab+a+b+1}}{(1-t)(1-t^a)(1-t^b)}.
\]
 \end{proposition}
 As a direct consequence of Proposition \ref{hilbertseries}, we get the following result.
\begin{corollary}\label{corhilbertseries}
The formula of the Hilbert series of $Y=\P(1,1,b)(\F_q)$ is given as follows.
\begin{equation}\label{11bhs}
HS_{Y}(t)=\frac{1-t^{q+1}-2t^{qb+1}+t^{qb+2}+t^{qb+q+1}}{(1-t)(1-t)(1-t^b)}.
\end{equation}
\end{corollary}

\begin{proofC} It follows from Theorem \ref{MS} and the graded free resolution given in Theorem \ref{1abfree} that we have \[\mathcal{K}_M(t)=\sum_{d\in \N W}(-1)^{0}\beta_{0,d}(M)t^d+\sum_{d\in \N W}(-1)^{1}\beta_{1,d}(M)t^d+\sum_{d\in \N W}(-1)^{2}\beta_{2,d}(M)t^d,\]
where, $\beta_{0,0}(F_0)=1$, $\beta_{1,\lambda_1}(F_1)=1$, $\beta_{1,\lambda_2}(F_1)=1$, $\beta_{1,\lambda_3}(F_1)=1$ and $\beta_{2,\sigma_1}(F_2)=1$, $\beta_{2,\sigma_2}(F_2)=1.$ Therefore, we get
\[\mathcal{K}_M(t)=1-t^{qa+1}-t^{qb+1}-t^{(q-1)ab+a+b}+t^{qa+qb+1}+t^{(q-1)ab+a+b+1}.\] So, we have the following rational function as a result of the Theorem \ref{MS},
\begin{eqnarray}
\quad \quad HS_{M}(t)=\frac{1-t^{qa+1}-t^{qb+1}-t^{(q-1)ab+a+b}+t^{qa+qb+1}+t^{(q-1)ab+a+b+1}}{(1-t)(1-t^a)(1-t^b)}.
\end{eqnarray}
which completes the proof.
\end{proofC}

Consider the formula of the Hilbert series of $\P(1,1,b)(\F_q)$ given in Corollary \ref{corhilbertseries}. Multiplying the sum $1-t^{q+1}-2t^{qb+1}+t^{qb+2}+t^{qb+q+1}$ by $1/(1-t)$ twice we get the following polynomial $p(t)$,
\begin{equation*}\label{11bsum1}
1+2t+\cdots+qt^{q-1}+(q+1)t^{q}+\cdots+(q+1)t^{qb}+(q-1)t^{qb+1}+\cdots+2t^{qb+q-2}+t^{qb+q-1}.
\end{equation*}
Multiplying this latter sum by $1/(1-t^b)$ we get the following Hilbert series $HS_{M}(t)$.

\begin{equation}\label{11bhssum}\sum\limits_{i=0}^{\infty}t^{ib}p(t)=\sum\limits_{i=0}^{\infty}t^{ib}\left(\sum\limits_{j=0}^{q}(j+1)t^{j}+\sum\limits_{k=q+1}^{qb}(q+1)t^{k}+\sum\limits_{s=1}^{q-1}(q-s)t^{qb+s}\right).
\end{equation}
Since the Hilbert series is of the form
\begin{equation}\label{eq:HilbSer} \sum\limits_{d\in \N W}H_{Y}(d)t^{d}=\sum\limits_{i=0}^{\infty}t^{ib}p(t),
\end{equation}
the value $H_{Y}(d)$ is the coefficient of $t^{d-ib}$ in the polynomial $p(t)$. So, we denote this coefficient by $C_i(d)$ and carefully implement this idea to get the following result. Notice that the coefficient of $t^j$ in $p(t)$ is
\begin{equation}\label{eq:coefficient}
\begin{cases}
j+1, &\mbox{ if } 0\le j\le q-1\\
q+1, &\mbox{ if } q\le j\le qb\\
qb+q-j, &\mbox{ if } qb+1 \le j\le qb+q-1.
\end{cases}
\end{equation}
\begin{theorem}\label{hilbertfunc1} Let $d=d_0b+r_0$ where $0\le r_0 <b$. If $q\le b$ then the Hilbert function of $Y=\P(1,1,b)(\F_q)$ is as follows
\begin{equation}\label{HS1} H_{Y}(d) =  \begin{cases}
d_0(q+1)+r_0+1, &\mbox{ if } 0\le d_0\le q-1 \mbox{ and } 0\le r_0\le q-1\\
(d_0+1)(q+1), &\mbox{ if } 0\le d_0\le q-1 \mbox{ and } q\le r_0 <b\\
q(q+1), &\mbox{ if } d_0\ge q \mbox{ and } 0<r_0< b\\
q(q+1)+1, &\mbox{ if } d_0\ge q \mbox{ and } r_0=0.
\end{cases}
\end{equation}
\end{theorem}
\begin{proof} Recall that $H_{Y}(d)$ is the coefficient of $t^d$ in the infinite sum given in \eqref{11bhssum}. So, we concentrate on the coefficient $C_i(d)$ of $t^{d-ib}$ in the polynomial $$p(t)=\sum\limits_{j=0}^{q}(j+1)t^{j}+\sum\limits_{k=q+1}^{qb}(q+1)t^{k}+\sum\limits_{s=1}^{q-1}(q-s)t^{qb+s}=\sum\limits_{i=0}^{d_0}C_i(d)t^{d-ib}$$
for $0\le i \le d_0.$ Thus, $H_{Y}(d)=\sum\limits_{i=0}^{d_0}C_i(d)$. Recall that $d-ib=(d_0-i)b+r_0$.

\textbf{Case I:} Let $0\le d_0\le q-1.$\\
For $i=d_0$, we have $d-ib=r_0$ and get the coefficient of $t^{r_0}$ in $p(t)$ to be
\begin{equation}C_{i}(d)=
\begin{cases}
r_{0}+1, &\mbox{ if } 0\le r_{0}\le q-1\\
q+1, &\mbox{ if } q\le r_0< b.
\end{cases}
\end{equation}
Whenever $0\le i \le d_0-1$, we have 
$q\le b \le (d_0-i)b+r_0 \le (d_0+1)b\le qb.$ 
So, $d-ib=(d_0-i)b+r_0$ is between $q$ and $qb$, and thus $C_{i}(d)=q+1$, for all $0\le i \le d_0-1$, by (\ref{eq:coefficient}).
Hence,
\begin{equation}
H_{Y}(d)=\sum_{i=0}^{d_0}C_i(d)=d_0(q+1)+\begin{cases}
r_{0}+1, &\mbox{ if } 0\le r_{0}\le q-1\\
q+1, &\mbox{ if } q\le r_0< b.
\end{cases}
\end{equation}
\textbf{Case II:} Let $d_0=q$. So, $d=qb+r_0$ and $d-ib=(q-i)b+r_0$. \\
Firstly, let us consider the case where $i=0$.
Then we get the following formula.
\begin{equation}
C_i(d)=\begin{cases}
q+1 &\mbox{ if } r_0=0,\\
q-r_0 &\mbox{ if } 0<r_0\le q-1,\\
0 &\mbox{ if } q\le r_0 < b.
\end{cases}
\end{equation}
Let us now consider the cases where $1\le i\le q-1$. Then,
\[q\le b\le (q-i)b\le(q-1)b \Rightarrow q\le (q-i)b+r_0\le (q-1)b+r_0<qb\Rightarrow C_{i}(d)=q+1.\]
For the case $i=q$, we get the following formula.
\begin{equation}
C_{i}(d)=\begin{cases}
1 &\mbox{ if } r_0=0,\\
r_0+1 &\mbox{ if } 0<r_0\le q-1,\\
q+1 &\mbox{ if } q\le r_0<b.
\end{cases}
\end{equation}
Thus, summing all these, we obtain
\begin{equation}
H_{Y}(d)=\begin{cases}
q(q+1)+1 &\mbox{ if } r_0=0,\\
q(q+1) &\mbox{ if } 0<r_0\le q-1,\\
q(q+1) &\mbox{ if } q\le r_0 < b.
\end{cases}
\end{equation}
\textbf{Case III:} Let $d_0>q$. So, $d_0-q>0$. \\
For the cases $0\le i <d_0-q$, we have $C_i(d)=0,$ since 
\[q+1\le d_0-i \Rightarrow (q+1)b\le (d_0-i)b\Rightarrow qb+(q-1)<  qb+b \le (d_0-i)b\le d-ib.\]
For the cases where $d_0-q\le i <d_0$, we have
\[0<d_0-i\le q \Rightarrow q\le b\le (d_0-i)b\le qb.
\]
So, the interval in which $d-ib=(d_0-i)b+r_0$ lies depends on $r_0$. \\
If $r_0=0$, we get $C_i(d)=q+1$ as $d-ib=(d_0-i)b$, for all $d_0-q\le i \le d_0-1$.\\
If $r_0>0$ and $i=d_0-q$, we get $d-ib=qb+r_0$ and so
\begin{equation}
C_i(d)=\begin{cases}
q-r_0 &\mbox{ if } 1\le r_0\le q-1,\\
0 &\mbox{ if } q\le r_0< b.
\end{cases}
\end{equation}
If $r_0>0$ and $i>d_0-q$, we get $d_0-i<q \Rightarrow d_0-i\le q-1$ and so
\[d-ib=(d_0-i)b+r_0\le (q-1)b+r_0=qb-b+r_0<qb.\]
So, we have $C_i(d)=q+1$, for all $d_0-(q-1)\le i \le d_0-1$.\\
Finally, for the case $i=d_0$, we have $d-ib=r_0$ and so
\begin{equation}
C_i(d)=\begin{cases}
r_0+1 &\mbox{ if } 0\le r_0\le q-1,\\
q+1 &\mbox{ if } q\le r_0<b.
\end{cases}
\end{equation}
Summarizing, we have $C_i(d)=q+1$, for all $d_0-(q-1)\le i \le d_0-1$ and that
\begin{equation}
C_{d_0-q}(d)+C_{d_0}(d)=\begin{cases}
(q+1)+1 &\mbox{ if }  r_0=0,\\
q+1 &\mbox{ if } 0 < r_0 < b.
\end{cases}
\end{equation}
Therefore, we obtain
\begin{equation}
H_{Y}(d)=\begin{cases}
q(q+1)+1 &\mbox{ if } r_0=0,\\
q(q+1) &\mbox{ if } 0<r_0 < b,
\end{cases}
\end{equation} 
completing the proof.
\end{proof}
\begin{theorem}\label{hilbertfunc2} Let $d=d_0b+r_0$ with $0\le r_0 <b$. When $d>q$, we also let $d=q+kb+r$ with $0\leq r <b$. If $q>b$ then the Hilbert function of $Y=\P(1,1,b)(\F_q)$ is as follows 
\begin{equation*}\label{HS2} H_{Y}(d) = \begin{cases}
d+1 &\mbox{ if } 0\le d\le b-1,\\
(d_0+1)(d+1-bd_0/{2})
&\mbox{ if } b\le d\le q,\\
(d_0-k)(d-(d_0+k+1)b/2)+q(k+1)+1 &\mbox{ if } q+1\le d\le qb\\
(k+q)(k-q+1)b/2+q(k+d+1)-dk-bd_0+\kappa &\mbox{ if } d>qb \mbox{ and } k<q,\\
q^2+r_0 +\kappa &\mbox{ if } d>qb \mbox{ and } k\ge q,

\end{cases}
\end{equation*}
where 
\begin{equation*} \kappa = \begin{cases}
q-r_0, &\mbox{ if } 0<r_0<q,\\
q+1 &\mbox{ if } r_0=0.
\end{cases}
\end{equation*}   
\end{theorem}
\begin{proof}
Let $d=d_0b+r_0$ for $0\le r_0< b<q.$ Recall that the Hilbert series in (\ref{eq:HilbSer}) is
\begin{equation}
\sum\limits_{i=0}^{\infty}t^{ib}\left(\sum\limits_{j=0}^{q}(j+1)t^{j}+\sum\limits_{k=q+1}^{qb}(q+1)t^{k}+\sum\limits_{s=1}^{q-1}(q-s)t^{qb+s}\right)=\sum\limits_{i=0}^{\infty}t^{ib}p(t)    
\end{equation} and that $C_i(d)$ is the coefficient of $t^{d-ib}$ in the polynomial $p(t)$ for $0\le i \le d_0$ as in the previous proof.
Putting $j=d-ib$ in (\ref{eq:coefficient}), we get
\begin{equation}\label{coeffpol}
C_i(d)=\begin{cases}
d-ib+1 &\mbox{ if } 0\le d-ib\le q,\\
q+1 &\mbox{ if } q+1 \le d-ib \le qb,\\
q+(qb-d+ib) &\mbox{ if } qb+1\le d-ib \le qb+q-1.
\end{cases}
\end{equation}

\textbf{Case I:} Let $0\le d\le b-1.$ So, we have $d_0=0$. Then $i=0$. Hence, $H_{Y}(d)=d+1.$

\textbf{Case II:} Let $b\le d\le q.$ It is clear that $d_0>0.$ And then $d-ib\le q.$ So, $C_i=d-ib+1.$ Hence, \[H_{Y}(d)=\sum\limits_{i=0}^{d_0}(d-ib+1)=(d_0+1)(d+1)-b\sum\limits_{i=0}^{d_0}i=(d_0+1)(d+1)-b\left(\frac{d_0(d_0+1)}{2}\right).\]

\textbf{Case III:} Let $q+1\le d< qb$ and also $d=q+kb+r$ with $0\leq r <b$. Since $d>q$ we have $d-q=kb+r$. Firstly, let $k=0.$ If $i=0$ we get $C_i=q+1.$
Let $i>0$. So, $d-d_0b<\dots<d-2b<d-b<q.$ And then we get $C_i=d-ib+1$ for $i=1,2,\dots,d_0.$ Hence,
$H_{Y}(d)=q+1+\sum\limits_{i=1}^{d_0}(d-ib+1).$

If $i=0$ we get $C_i=q+1.$ Let $1\le i \le k.$ It is clear that $q\le d-b < (q-1)a$ so we have $q\le d-kb\le d-ib \le d\le qb$. Then, we get $C_i=q+1$ for $1\le i\le k.$ Let $k+1\le i\le d_0.$ Since $-i\le -k-1$ we have $d-ib\le d-kb-b$. Also since $d=q+kb+r$ and $r<b$ we get $d-ib\le d-kb-b\le q+r-b <q.$ Hence,
$H_{Y}(d)=(k+1)(q+1)+\sum\limits_{i=k+1}^{d_0}(d-ib+1).$
If we arrange the equality we get,
\[H_{Y}(d)=(k+1)(q+1)+(d_0-k)(d+1)-b\left(\frac{d_0(d_0+1)}{2}-\frac{k(k+1)}{2}\right).\]

\textbf{Case IV:} Let $qb<d$ and $k< q.$  Since $d=d_0b+r_0$ we have $d_0\ge q.$ 

Let $0\le i < d_0-q.$ If we consider $i< d_0-q,$ we get $d-ib>qb.$ So, we get, \[C_i=q+qb-d+ib.\]

Let $d_0-q<i\le d_0-k.$ If we consider $d_0-q< i$, using the argument in Case V, we get $d-ib<qb.$ And also if we consider $i\le d_0-k,$ we have $d-ib\ge d-d_0b+kb=r_0+kb>kb=d-q-r>qb-q-r>qb-2q=q(b-2).$ Since $b\ge 2$ we have $d-ib>q(b-2)\ge q.$ Then we get, \[C_i=q+1.\]

Let $d_0-k <i \le k.$ If we consider $d_0-k<i$ we have $d-ib\le d-d_0b+kb-b=r_0-b+kb<qb+r_0-b.$ Since $r_0<b$ we get $d-ib<qb.$ And if we consider $i\le k$ we have $d-ib\ge d-kb=q+r>q.$ So we get, \[C_i=q+1.\]

Let $k+1\le i \le d_0.$ Using the argument in Case V, we get, \[C_i=d-ib+1.\]
Let consider $i=d_0-q$. If $i=d_0-q$ we have $d-ib=d-d_0b+qb=r_0+qb.$ Since $0<r_0<q$ we have $qb< d-ib\le qb+q-1$. So, we get \[C_i=q-r_0 \mbox{ for } 0<r_0<q.\] Also, if we consider $r_0=0$ we get $d-ib=d-d_0b+qb=r_0+qb=qb$. So, \[C_i=q+1 \mbox{ for } r_0=0.\] 
Therefore we get,

\[H_{Y}(d)=\sum\limits_{i=0}^{d_0-q-1}(q+qb-d+ib)+(q+k-d_0)(q+1)+\sum\limits_{i=k+1}^{d_0}(d+1-ib)+\kappa\] where \begin{equation*} \kappa = \begin{cases}
q-r_0, &\mbox{ if } 0<r_0<q,\\
q+1 &\mbox{ if } r_0=0.
\end{cases}
\end{equation*}
If we arrange all equalities for $d > qb$ and $k < q$ we get,
\begin{align*}
H_{Y}(d)&= (d_0-q)(q+qb-d)+\frac{b}{2}\left((d_0-q-1)(d_0-q)\right)+(q+k-d_0)(q+1)\\
&+(d+1)(d_0-k)-\frac{b}{2}\left(d_0(d_0+1)-k(k+1)\right)+\kappa.\\
&=qd-\frac{q^2b}{2}-bd_0+\frac{bq}{2}+q+kq-dk+\frac{bk^2}{2}+\frac{bk}{2}+\kappa\\
&=\frac{b}{2}\left((k+q)(k-q+1)\right)+q(d+1+k)-bd_0-kd+\kappa  \begin{cases}
q-r_0, &\mbox{ if } r_0<q,\\
q+1 &\mbox{ if } r_0=0.
\end{cases}
\end{align*}

\textbf{Case V:} Let $qb<d$ and $k\ge q.$ Recall that $d=d_0b+r_0$ and we have $d_0\ge q.$ 

Let $0\le i \le k-q.$ Then, we get $d-ib\ge d-kb+qb =q+r+qb\ge qb+q>qb+q-1$. Since $d-ib>qb+q-1$ we get \[C_i=0.\]

Let $k-q <i <d_0-q.$ Firstly if we consider $i<d_0-q$ we get $-i>q-d_0\Rightarrow d-ib>d+qb-d_0b=r_0+qb>qb.$
Also, if we consider $i>k-q$ we get $-i\le -k+q-1\Rightarrow d-ib \le d-kb+qb-b.$ Since $d=q+kb+r$ and $r<b$ we get $d-ib \le d-kb+qb-b \le q+r+qb-b\le qb+q-1.$ So, we get $qb<d-ib\le qb+q-1$ and then we get \[C_i=q+qb-d+ib.\]

Let $d_0-q<i\le k.$ Then since $d=kb+q+r$ we have $d-ia\ge d-kb =q+r\ge q.$ Also, we have $i\ge d_0-q+1\Rightarrow -i\le q-1-d_0\Rightarrow d-ib\le d+qb-b-d_0b=qb+r_0-b.$ Since $r_0<b$ we get $d-ib<qb.$ So,
\[C_i=q+1.\]

Let $k+1\le i\le d_0.$ Using the argument in Case IV, we get \[C_i=d-ib+1.\] 
Let $i=d_0-q$ as in \text{Case IV}. Then, we have $d-ib=d-d_0b+qb=r_0+qb.$ Since $0<r_0<q$ we have $qb< d-ib\le qb+q-1$. So, we get \[C_i=q-r_0 \mbox{ for } 0<r_0<q.\] Also, if we consider $r_0=0$ we get $d-ib=d-d_0b+qb=r_0+qb=qb$. So, \[C_i=q+1 \mbox{ for } r_0=0.\] 

Therefore if we arrange all equalities for $d>qb$ and $k\ge q$ we get the desired statement.
\end{proof}

Based on Theorem \ref{hilbertfunc1} and Theorem \ref{hilbertfunc2}, we describe the regularity set of the $\F_q$-rational points $Y=\P(1,1,b)(\F_q)$ in the next consequence. 

\begin{corollary}\label{cor:reg}
Let $b$ be a positive integer. The regularity set of $Y=\P(1,1,b)(\F_q)$ is given as follows. 
\[\reg(Y)=\{d_0b \in \N:  d_0\ge (q+\lfloor (q-1)/b\rfloor)\}=(q+\lfloor (q-1)/b\rfloor)b+\N b.\]
\end{corollary}

\begin{proof}
$d=d_0b+r_0$ with $0\le r_0 <b$ as before.

\textbf{Case A:}
Firstly let $q\le b.$ For this case, it is sufficient to show the following equality since $\lfloor (q-1)/b\rfloor=0$.
\[\reg(Y)=\{d \in \N: d=d_0b \mbox{ with } d_0\ge q\}=qb +\N b.\]
\textbf{Case A.I:}  Let $0\le d_0\le q-1$ and $0\le r_0\le q-1.$ If we consider the result given in \ref{hilbertfunc1} we get the value of Hilbert function as $H_{Y}(d)=d_0(q+1)+r_0+1.$ So,
\[H_{Y}(d)\le (q-1)(q+1)+q=q^2-1+q<q^2+q+1.\] Therefore we have $d\not\in \reg(Y)$ for $0\le d_0\le q-1$ and $0\le r_0\le q-1$.

\textbf{Case A.II:} Let $0\le d_0$ and $q\le r_0<b.$ Similarly, if we consider the result given in \ref{hilbertfunc1} we get $H_{Y}(d)=(d_0+1)(q+1)\le q(q+1).$ Then it is clear that $H_{Y}(d)<q^2+q+1.$ So, we get $d\notin \reg(Y)$.

\textbf{Case A.III:} Let $d_0\ge q$ and $0<r_0<b$. We get $H_{Y}(d)=q(q+1)$ from the result given in \ref{hilbertfunc1}. It is clear that $H_{Y}(d)<q^2+q+1.$ Therefore, we get $d\notin \reg(Y)$ for $d_0\ge q$ and $0<r_0<b$.

\textbf{Case A.IV:} Let $d_0\ge q$ and $r_0=0$. Since the result given in \ref{hilbertfunc1} we know that $H_{Y}(d)=q^2+q+1$ for $d_0\ge q$ and $r_0=0$. So, $d=d_0b\in \reg(Y)$ in this case. 

Since, only one of the situations above is possible, it follows that we have
\[\reg(Y)=\{d\in \N: d=d_0b \mbox{ with } d_0\ge q\}=qb+\N b.\]

\textbf{Case B:} Let $q>b$. 

\textbf{Case B.I:} Let $0\le d\le b-1.$ We get $H_{Y}(d)=d+1$ from the result given in \ref{hilbertfunc2}. It is clear that $H_{Y}(d)<q^2+q+1.$ Therefore, we get $d\notin \reg(Y)$ for $0\le d\le b-1$.

\textbf{Case B.II:} Let $b\le d\le q$. We have $d_0\ge 1$, as otherwise we get the contradiction $ d=d_0b+r_0=r_0 <b$. Then, we get the following inequality,
\[H_{Y}(d)=(d_0+1)(d+1)-b\left(\frac{d_0(d_0+1)}{2}\right)< (d_0+1)(d+1)\le q(q+1).   \] Therefore we get $H_{Y}(d)<q^2+q+1.$ So, we get $d\notin \reg(Y)$ for $b\le d\le q$.

\textbf{Case B.III:} Let $q+1\le d\le qb$. Consider the following sum:
\[H_{Y}(d)=(k+1)(q+1)+(d_0-k)(d+1)-b\sum\limits_{i=k+1}^{d_0}i.\] If we arrange this sum we get,
\begin{align*}
&=kq+q+1+d_0d+d_0-kd-\frac{bd_0^2}{2}-\frac{bd_0}{2}+\frac{bk^2}{2}+\frac{bk}{2}\\
&=kq+q+1+bd_0^{2}+d_{0}r_{0}+d_0-kd-\frac{bd_0^{2}}{2}-\frac{bd_0}{2}+\frac{k^2b}{2}+\frac{kb}{2}\\
&<kq+q+1+\frac{bd_0^2}{2}+\frac{bd_0}{2}+d_0-kd+\frac{k^2b}{2}+\frac{kb}{2}\\
&=kq+q+1+\frac{bd_0^2}{2}+\frac{bd_0}{2}+d_0-kq-k^2b-kr+\frac{k^2b}{2}+\frac{kb}{2}\\
&=q+1+\frac{bd_0^{2}}{2}+\frac{bd_0}{2}+d_0-\frac{k^2b}{2}-kr+\frac{kb}{2}.
\end{align*}

Since $k>1\Rightarrow -k<-1$, $d=q+kb+r\le qb \Rightarrow k\le q-\frac{q+r}{b},$  we get the following inequalities.


\begin{align*}
&<q+1+\frac{bd_0^{2}}{2}+\frac{bd_0}{2}+d_0+\frac{b}{2}-r+\frac{kb}{2}\\
&=q+1+d_0+\frac{b}{2}\left(d_0(d_0+1)\right)+\frac{b}{2}\left(k+1\right)-r\\
&\le q+1+d_0+\frac{bd_0}{2}\left(d_0+1\right)+\frac{b}{2}\left(q-\frac{q+r}{b}+1\right)-r\\
&=q+1+d_0+\frac{bd_0}{2}\left(d_0+1\right)+\frac{qb}{2}-\frac{qb}{2b}-\frac{rb}{2b}+\frac{b}{2}-r\\
&=\frac{q}{2}+1+d_0+\frac{bd_0}{2}\left(d_0+1\right)+\frac{qb}{2}+\frac{b}{2}-\frac{3r}{2}.
\end{align*}
As we know that $q+1\le d=d_0b+r_0$, we have $q\le d_0b+r_0-1.$ And also, \[d_0\le q\le d_0b+r_0-1\Rightarrow d_0\le d_0b+r_0-1\Rightarrow d_0\le \frac{r_0-1}{1-b}.\] If we use all these inequalities we get,
\begin{align*}
&\le \frac{bd_0+r_0}{2}+\frac{1}{2}+\frac{r_0-1}{1-b}+\frac{bd_0}{2}\left(\frac{r_0-1+1-b}{1-b}\right)+\frac{qb}{2}+\frac{b}{2}-\frac{3r}{2}\\
&< \frac{bd_0}{2}+b+\frac{1}{2}+\frac{b-1}{1-b}+\frac{qb}{2}-\frac{3r}{2}\\
&=\frac{bd_0}{2}+b-\frac{1}{2}+\frac{qb}{2}-\frac{3r}{2}\\
&\le qb+b-\frac{(3r+1)}{2}\\
&<q^2+q<q^2+q+1.
\end{align*}
Therefore we get $H_{Y}(d)<q^2+q+1.$ So, $d\notin \reg(Y)$ for $q+1\le d\le qb$.

\textbf{Case B.IV:} Let $d>qb$ and $k<q.$
Consider the following sum:
\[H_{Y}(d)=\frac{b}{2}\left((k+q)(k-q+1)\right)+q(k+d+1)-dk-bd_0+\kappa.\]
The formula above for $H_{Y}(d)$ depends on the value of $k$ in the formula $d=q+kb+r$. When $k=q-1$, we get
\begin{align*}
H_{Y}(d)&=\frac{b}{2}\left((q-1+q)(q-1-q+1)\right)+q(q-1+d+1)-qd+d-bd_0+\kappa\\
&= q^2+d-bd_0+\kappa=q^2+r_0+\kappa.
\end{align*}
\begin{itemize}
    \item If $0<r_0<q$, then $\kappa=q-r_0$ and so $H_{Y}(d)=q^2+q<q^2+q+1$.
Thus, we get $d\notin \reg(Y)$ for $d>qb,$ $k=q-1$ and $0<r_0<q.$
\item Let $r_0=0.$ Then, $\kappa=q+1$ and so $H_{Y}(d)=q^2+q+1$. Therefore, $d\in \reg(Y)$ for  $d>qb,$ $k=q-1$ and $r_0=0$. 
\end{itemize}
Since, $H_{Y}(q+kb+r) < H_{Y}(q+(q-1)b+r)$ for $k<q-1$,it follows that $d\notin \reg(Y)$ if $k<q-1$. 

\textbf{Case B.V:} Let $d>qb$ and $k\ge q.$
\begin{itemize}
\item Let $0<r_0<q.$ So, $\kappa=q-r_0.$ Then, we have 
\[H_{Y}(d)=q^2+d-bd_0+\kappa=q^2+q<q^2+q+1.\] 
So, we get $d\notin \reg(Y)$ for $d>qb, k\ge q$ and $0<r_0<q.$
\item Let $r_0=0.$ So, $\kappa=q+1.$ Then, we have
\[H_{Y}(d)=q^2+d-bd_0+\kappa=q^2+q+1.\] 
So, $d\in \reg(Y)$ in this case. 
\end{itemize}

All these cases reveal that $d\in \reg(Y)$ if and only if  $d=d_0b$, $d_0>q$ and $k\ge q-1$. In order to determine the smallest element in $\reg(Y)$, we take $k=q-1$ and look for the value of $r$ for which $d=d_0b\in \reg(Y)$. Since $0\le r \le b-1$, we have 
\[
q+(q-1)b\le q+kb+r=d \le q+qb-b+b-1=qb+q-1.
\]
Thus, we get $qb+q-b \le d=d_0b \le qb+q-1$. As there are $b$ integers in this interval, the unique integer in the following interval 
\[
q+(q-b)/b\le d_0\le q+(q-1)/b
\]
is $d_0=q+\lfloor q-1/b\rfloor$. Therefore,
\begin{align*}
\reg(Y)&=\{d\in \N: d=d_0b=q+kb+r \mbox{ with } d_0\le q+(q-1)/b, d_0> q \mbox{ and } k\ge q-1\}\\
&=(q+\lfloor q-1/b\rfloor)b+\N b    
\end{align*}
\end{proof}


\begin{theorem}\label{serre}\cite[Theorem 4.3.5]{BJ1993} Let $R$ be a positively graded $\K$-algebra and $M\neq 0$ a finite(a.k.a. finitely generated) graded $R$-module. Then,
\begin{enumerate}[label=\roman*.]
    \item there exist a uniquely determined quasi-polynomial $P_{M}$ with $H_{M}(d)=P_{M}(d)$ for all $d\gg 0.$
    \item one has $\deg{(HS_{M}(t))}=\max\{d:H_{M}(d)\neq P_{M}(d)\}$
\end{enumerate}
Here, $\deg(H_{M}(t))$ denotes the degree of the rational function $HS_{M}(t)$ and is also known as the $a$-invariant of $M$ denoted $a(M)$.
\end{theorem}

\begin{remark} It follows from Theorem \ref{serre} that the Hilbert function of $Y$ agrees with the Hilbert \textbf{quasi-polynomial} $P_Y$, for all $d>a(Y)$, i.e. there exist a positive integer $g$ (period) and the polynomials $P_{0},\dots,P_{g-1}$ such that $H_{Y}(d)=P_{i}(d)$ for $d>a(Y)$ and $d\equiv i\mod g.$
\end{remark}

As a direct consequence of the above results, we give now the Hilbert quasi-polynomial of $Y$.
\begin{corollary} Let $Y=\P(1,1,b)( \F_q)$. Then the Hilbert quasi-polynomial of $Y$ is given by 
\begin{align*}\label{reg:unit} 
P_{Y}(d_0b+r_0)=\begin{cases}
q(q+1)+1 &\mbox{ if } d_0 \in \Z \mbox{ and } r_0=0\\
q(q+1) &\mbox{ if } d_0 \in \Z \mbox{ and } 1\le r_0 \le b-1.
\end{cases}
\end{align*}
\end{corollary}
\begin{proof} We need to prove that $H_Y(d)=P_Y(d)$ for all 
$$d\ge a(Y)+1=(q-1)(b+1)+1=(q-1)b+q=qb+q-b.$$ 
The proof of Corollary \ref{cor:reg} reveals that $H_Y(d)<q(q+1)$ for all $d\le a(Y)$, and $H_Y(d)=q(q+1)$ if and only if $d=(q-1)b+r_0$ with $q\le r_0 <b$ or $d=d_0b+r_0$ with $d_0\ge q$ and $0 < r_0 <b$, in Case A where we have $q\le b$. Finally, if $d=d_0b$ with $d_0\ge q$, $H_Y(d)=q(q+1)+1$, proving the assertion for Case A.  
As for Case B where we have $q>b$, we get $d\ge a(Y)+1=qb+q-b>qb$. Thus, $H_{Y}(d)=q^2+q$ as the proof of the Corollary \ref{cor:reg} shows, when $0<r_0<q$ and $d>qb$. Similarly, $H_Y(d_0b)=q(q+1)+1$ if $d_0>q$. Therefore, we have $H_Y(d)=P_Y(d)$ in this case also as desired.
\end{proof}


Since the relation between the dimension of codes over $Y$ and the Hilbert function of $Y$ given in \eqref{dimhil} we get the following result.
\begin{corollary}\label{dimcor1}Let $d=d_0b+r_0$ where $0\le r_0 <b$ and $Y=\P(1,1,b)(\F_q)$. If $q\le b$, then we get the dimension of the code $C_{d,Y}$ as follows.
\begin{equation}\label{dim1} \dim_{\F_q}(C_{d,Y})  = \begin{cases}
d_0(q+1)+r_0+1 &\mbox{ if } 0\le d_0\le q-1 \mbox{ and } 0\le r_0\le q-1\\
(d_0+1)(q+1) &\mbox{ if } 0\le d_0\le q-1 \mbox{ and } q\le r_0 <b\\
q(q+1) &\mbox{ if } d_0\ge q \mbox{ and } 0<r_0< b\\
q(q+1)+1 &\mbox{ if } d_0\ge q \mbox{ and } r_0=0
\end{cases}
\end{equation}
\end{corollary}
\begin{corollary}\label{dimcor2}Let $d=d_0b+r_0$ with $0\le r_0 <b$ for all the cases and $d=q+kb+r$ with $0\leq r <b$ for the cases where $d>q$. If $q>b$ then 
\begin{equation*}\label{dim2} \dim_{\F_q}(C_{d,Y})= \begin{cases}
d+1 &\mbox{ if } 0\le d\le b-1,\\
(d_0+1)(d+1-bd_0/{2})
&\mbox{ if } b\le d\le q,\\
(d_0-k)(d-(d_0+k+1)b/2)+q(k+1)+1 &\mbox{ if } q+1\le d\le qb\\
(k+q)(k-q+1)b/2+q(k+d+1)-dk-bd_0+\kappa &\mbox{ if } d>qb \mbox{, } k<q,\\
q^2+r_0 +\kappa &\mbox{ if } d>qb \mbox{, } k\ge q,


\end{cases}
\end{equation*}
where 
\begin{equation*} \kappa = \begin{cases}
q-r_0, &\mbox{ if } 0<r_0<q,\\
q+1 &\mbox{ if } r_0=0.
\end{cases}
\end{equation*}   
\end{corollary}

\subsection{Examples for Regularity Set and Dimension}\label{sec:exreg}
In this part, the values of the Hilbert function of $Y=\P(1,1,b)(\F_q)$ which is one of the main results given in this article will be presented. In the tables given in this part, we emphasize the first element of the regularity set which is the set of degrees in which the Hilbert function reaches its maximum value. 
\begin{example}
Let $q=5$, $X=\P(1,1,b)$ and $Y=X(\F_q)$. We will give the values of the Hilbert Function of $Y$ in the following tables for $b=2,5,7$, respectively. And by Corollary \ref{dimcor1} and Corollary \ref{dimcor2} these will be the dimensions of codes over $Y$. We note that these values were calculated using Macaulay2 \cite{M2}.
\begin{table}[H]
\addtolength{\tabcolsep}{-3pt}
\centering
\caption{The values of Hilbert Function of $\P(1,1,2)$ over $\F_{5}$}
\vskip-1.2em
\begin{tabular}{|l|l|l|l|l|l|l|l|l|l|l|l|l|l|l|l}\hline
$d$     & 0 & 1 & 2 & 3 & 4 & 5 & 6 & 7 & 8 & 9 & 10 & 11 & 12   \\ \hline
$H_{Y}(d)$ & 1 & 2 & 4 & 6& 9& 12 & 15 & 18 &21 & 24 & 27& 28 & 30    \\ \hline
$d$& 13 & {\color{darkred}\textbf{14}} & 15 & 16 & 17 & 18 & 19 & 20 & 21 & 22 &23 &24 &25\\ \hline
$H_{Y}(d)$& 30& 31 & 30 & 31 & 30 & 31 & 30 & 31 & 30 &31 & 30 & 31 &30 \\ \hline
\end{tabular}\end{table}
\vskip-1.6em
\begin{table}[H]
\addtolength{\tabcolsep}{-3pt}
\centering
\caption{The values of Hilbert Function of $\P(1,1,5)$ over $\F_{5}$}
\vskip-1.2em
\begin{tabular}{|l|l|l|l|l|l|l|l|l|l|l|l|l|l|l|l}
\hline
$d$     & 0 & 1 & 2 & 3 & 4 & 5 & 6 & 7 & 8 & 9 & 10 & 11 & 12   \\ \hline
$H_{Y}(d)$ & 1 & 2 & 3 & 4 & 5 & 7 & 8 & 9 & 10 & 11 & 13 & 14& 15    \\ \hline
$d$& 13 & 14 & 15 & 16 & 17 & 18 & 19 & 20 & 21 & 22 &23 &24 &{\color{darkred}\textbf{25}}\\ \hline
$H_{Y}(d)$& 16& 17 & 19 & 20 & 21 & 22 & 23 & 25 & 26 & 27 & 28 & 29 & 31 \\ \hline
$d$& 26 & 27 & 28 & 29 & 30 & 31 & 32 & 33 & 34 & 35 & 36 & 37 & 38\\ \hline
$H_{Y}(d)$& 30& 30 & 30& 30 & 31 & 30 & 30 & 30 & 30 & 31 & 30 & 30 & 30 \\ \hline
\end{tabular}
\end{table}
\vskip-1.6em
\begin{table}[H]
\addtolength{\tabcolsep}{-3pt}
\centering
\caption{The values of Hilbert Function of $\P(1,1,7)$ over $\F_{5}$}
\vskip-1.2em
\begin{tabular}{|l|l|l|l|l|l|l|l|l|l|l|l|l|l|l|l}
\hline
$d$     & 0 & 1 & 2 & 3 & 4 & 5 & 6 & 7 & 8 & 9 & 10 & 11 & 12   \\ \hline
$H_{Y}(d)$ & 1 & 2 & 3 & 4 & 5 & 6 & 6 & 7 & 8 & 9 & 10 & 11& 12    \\ \hline
$d$& 13 & 14 & 15 & 16 & 17 & 18 & 19 & 20 & 21 & 22 &23 &24 &25\\ \hline
$H_{Y}(d)$& 12& 13 & 14 & 15 & 16 & 17 & 18 & 18 & 19 & 20 & 21 & 22 & 23 \\ \hline
$d$& 26 & 27 & 28 & 29 & 30 & 31 & 32 & 33 & 34 & {\color{darkred}\textbf{35}} & 36 & 37 & 38\\ \hline
$H_{Y}(d)$& 24& 24 & 25& 26 & 27 & 28 & 29 & 30 & 30 & 31 & 30 & 30 & 30 \\ \hline
$d$& 39 & 40 & 41 & 42 & 43 & 44 & 45 & 46 & 47 & 48 & 49 & 50 & 51\\ \hline
$H_{Y}(d)$& 30& 30& 30& 31 & 30 & 30 & 30 & 30 & 30 & 30 & 31 & 30 & 30 \\ \hline
\end{tabular}
\end{table}
\end{example}

\begin{remark}
We note that in these tables, we emphasize the first element $d$ of the regularity set in boldface red.   
\end{remark}

\section{Minimum Distance}\label{sec:minimum}
Recall that $Y=\{[1:y_2:y_3] : y_2,y_3\in \F_q\} \cup \{[0:y_2:1] : y_2\in \F_q\}\cup \{[0:1:0]\}$ is the set of $\F_q$-rational points of the weighted projective space $X=\P(1,1,b)$ over the algebraically closed field $\overline{\F}_q$, where $b$ is a positive integer.

\begin{lemma}\label{lm:mdis:1}
If $0<d<b$ and then $\delta(C_{d,Y}) \ge q$. 
\end{lemma}
\begin{proof} Let $F\in S_d\setminus \{0\}$. Since $\deg(x_3)=b$ and $d<b$, it is clear that $F\in \F_q[x_1,x_2]$. So, we assume that $F=x_{1}^{\ell}F'(x_1,x_2)$ where $F'\in \F_q[x_1,x_2]$ is a homogeneous polynomial of degree $d-\ell$, not divisible by $x_1$. Thus, $F'(x_1,x_2)=x_1F_1(x_1,x_2)+F_2(x_2)$ where $F_1$ and $F_2$ are homogeneous polynomials of degrees $d-\ell-1$ and $d-\ell$, respectively. So, $F_2(x_2)=cx_2^{d-\ell}$ for some $c\in \F_q^*$. Therefore,
\begin{equation}\label{e:TwoVariableF}
    F(x_1,x_2,x_3)=x_{1}^{\ell}[x_1F_1(x_1,x_2)+cx_2^{d-\ell}].
\end{equation}

\textbf{Case I:} Let $\ell>0$. Then, $F$ has $q+1$ roots $[y_1:y_2:y_3]$ with $y_1=0$. Since $f(x_2):=F(1,x_2,x_3)\in \F_q[x_2]\setminus \{0\}$ is univariate, it can have at most $q$ roots $y_2\in \F_q$. If $f$ has $q$ roots then $F$ vanishes at $q^2+q+1$ points in $Y$ and hence $F\in I(Y)$. Therefore, the codeword  $\ev_{d,Y}(F)=0$. Thus, in order to get a non-trivial codeword $f$ may have at most $q-1$ roots and then $F$ has at most $(q-1)q$ roots with $y_1=1$, $y_2\in V(f)\subset \F_q$ and $y_3\in \F_q$. Hence, $F$ can have at most $q+1+(q-1)q=q^2+1$ roots in $Y$.

\textbf{Case II:}
Let $\ell=0$.  If  $F(0,y_2,y_3)=0$ we get $y_2=0$ from (\ref{e:TwoVariableF}). So, $F$ has only one root $[0:0:1]$ in $Y$ with $y_1=0$.

In order to study the roots where $y_1=1$, we consider the univariate polynomial $f(x_2)=F(1,x_2,x_3)$ again. Since $f(x_2)$ has at most $q$ roots $y_2\in \F_q$, $F$ can have at most $q^2$ roots $[1:y_2:y_3]$ in $Y$. In total $F$ can have at most $q^2+1$ roots in $Y$. 

Therefore, a codeword $\ev_{d,Y}(F)=(F(P_1),\dots,F(P_N))$, where $N=q^2+q+1$, can have at least $q=N-(q^2+1)$ non-zero components in both cases, yielding $\delta(C_{d,Y}) \ge q$.
\end{proof}
\begin{theorem} \label{t:mindis}
The minimum distance of the code $C_{d,Y}$ is given by 
\begin{equation}\label{dis1} \delta=  \begin{cases}
q &\mbox{ if } q\le d < b\\
q(q-d+1)&\mbox{ if } d< q\le b\\
q(q-d+1)&\mbox{ if } d < b< q\\
q(q-d+1)&\mbox{ if } b\leq d <q\\
q-k&\mbox{ if } b\le q\leq d, d=q+kb+r \mbox{ with } 0\le r<b \mbox{ and } 0\le k \le q-2\\
q-k&\mbox{ if } q< b \le d, d=q+kb+r \mbox{ with } 0\le r<b \mbox{ and } 0\le k \le q-2\\
1 & \mbox{ if } d=q+kb+r \mbox{ with } 0\le r<b \mbox{ and } k \ge q-1
\end{cases}
\end{equation}
\end{theorem}
\begin{proof}
\textbf{Case I:} $\mathbf{q\le d<b:}$
The polynomial $$F_0=x_1 x_2^{d-q} \prod_{y_2\in \F_q^*} (x_2-y_2x_1) \in S_d,$$ have $q+1$ roots $[y_1:y_2:y_3]$ with $y_1=0$ and $q(q-1)$ roots with $y_1=1$ in $X$. In total, $F$ has $q^2+1$ roots in $Y$. We know that $N=q^2+q+1$. Therefore, we get $\delta \le N-(q^2+1)=q.$ So, we get $\delta=q$ by Lemma \ref{lm:mdis:1}.

\textbf{Case II:} $\mathbf{d<q\le b:}$
For a non-zero $F\in S_d$, we set $F=x_{1}^{\ell}F'(x_1,x_2)$ where $F'$ is a homogeneous polynomial of degree $d-\ell$ of the form $F'(x_1,x_2)=x_1F'_1+cx_2^{d-\ell}$ with $c\in \F_q^*$, as in the proof of Lemma \ref{lm:mdis:1}. Notice that the univariate polynomial $f(x_2)=F(1,x_2)$ has degree $d-\ell$ and can have at most $d-\ell$ roots $y_2\in \F_q$. 

If $\ell >0$, then $F$ will have $q+1$ roots with $y_1=0$ and have at most $q(d-\ell)$ roots with $y_1=1$. Altogether, $F$ can have at most $qd+1+q(1-\ell)\leq qd+1$ roots, since $\ell\geq 1$. 

If $\ell=0$, then $F$ will have $1$ root with $y_1=0$ and have at most $qd$ roots with $y_1=1$. Altogether, $F$ can have at most $qd+1$ roots. 

The polynomial $$F_0=x_1 \prod_{y_2=1}^{d-1} (x_2-y_2x_1) \in S_d,$$ have $q+1$ roots $[y_1:y_2:y_3]$ with $y_1=0$ and $q(d-1)$ roots with $y_1=1$ in $X$. Thus, it has $qd+1$ roots in $Y$. Therefore, $\delta=q^2+q+1-(qd+1)=q(q+1-d)$.

\textbf{Case III:} $\mathbf{d<b<q:}$ The argument used in the proof of the previous situation applies here and we get $\delta=q(q-d+1).$

\textbf{Case IV:} $\mathbf{b\leq d<q:}$ Let $d=d_0b+r_0$ where $0\le r_0 <b$. For a non-zero $F\in S_d$, we consider the following key subset
$$J=\{y_3\in \F_q: y_3x_1^b-x_3 \mbox{ divides }F\}.$$
It is clear that $|J|\leq d_0$. It follows that $F(1,y_2,y_3)=0$ whenever $y_3\in J$ and there are $q|J|$ such points in $Y$ with $y_1=1$. 

On the other hand, the polynomial $f(x_2)=F(1,x_2,y_3)\in \F_q[x_2]\backslash \{0\}$ if $y_3\notin J$. This is because, in general, we have $F(x_1,x_2,x_3)=(y_3x_1^b-x_3)H+r(x_1,x_2)$ for some $H\in \F_q[x_1,x_2,x_3]$ and non-zero homogeneous polynomial $$r(x_1,x_2)=\sum\limits_{i=0}^{d}r_ix_1^{i}x_2^{d-i}$$ of degree $d$. So, if $f\equiv 0$, then $r(1,x_2)=\sum\limits_{i=0}^{d}r_ix_2^{d-i}$ is a zero polynomial, i.e. $r_i=0$ for all $i$, meaning that $r=0$ as a polynomial, a contradiction.  So, if $y_3 \notin J$ then $f$ has at most $\deg_{x_2}(F)=\deg (f)$ many roots. Therefore $F$ has $|\F_q\backslash J|\deg_{x_2}(F)=(q-|J|)\deg_{x_2}(F)$ many such roots at most. Thus, we have
\begin{equation}\label{eq:ub1}
|V_{Y}(F)\cap U_1|\le q|J|+(q-|J|)\deg_{x_2}(F),
\end{equation}
where $V_{Y}(F)=\{P\in Y : F(P)=0\}$ and $U_1=\{[x_1:x_2:x_3] \in Y : x_1=1\}$. 

Consider now the following general description of a homogeneous polynomial of degree $d=d_0b+r_0$, where $0\leq r_0 <b$:
$$F(x_1,x_2,x_3)=x_1^{\ell}\prod\limits_{y_3\in J}^{}(y_3x_1^b-x_3) F'(x_1,x_2,x_3) \text { where } F'=x_1F_1+F_2,$$
and $F_2(x_2,x_3)$ is a homogeneous polynomial of degree $d-\ell-|J|b$ with $x_1 \nmid F_2$. Let us estimate the roots of $F$ with $x_1=0$.

If $\ell>0$, then there are $q+1$ roots of $F$ with $x_1=0$. Thus, we have
\begin{eqnarray}
|V_{Y}(F)| &\le& q+1+q|J|+(q-|J|)\deg_{x_2}(F) \nonumber \\
&\leq& q+1+q|J|+(q-|J|)(d-\ell-|J|b). \nonumber\\
&=& q+1+q(d-\ell)+|J|(q-qb-d+\ell+|J|b) \nonumber\\
&\leq& q+1+q(d-\ell) \leq qd+1, \nonumber
\end{eqnarray}
since as we prove now $q-qb-d+\ell+|J|b\leq 0$ is satisfied:
if $b=1$, then we have $$q-qb-d+\ell+|J|b=\ell-d+|J|\leq 0;$$ 
and if $b\geq 2$, then $q+\ell \leq 2q \leq qb$ yielding the following
$$q-qb-d+\ell+|J|b \leq q-qb-d+\ell+d_0b \leq q-qb+\ell \leq 0.$$
 
If $\ell=0$, then $F(0,y_2,1)=0$ implies $F'(0,y_2,1)=0$, i.e. $F'\in  I([0:y_2:1])$. By \cite[Proposition 3.4]{MS2022}, we have $$x_1F_1+F_2=F' \in I([0:y_2:1])=\langle x_1, x_2^{b}-y_2^{b}x_3 \rangle$$ and thus $x_2^{b}-y_2^{b}x_3$ is a factor of $F_2$, when $y_2\neq 0$. Furthermore, $F_2$ can have at most $d_0-|J|$ such factors, since there are at most $d_0b$ in $d$, and $\deg_{x_2}(F_2)=d-|J|b$. Hence, there are at most $d_0-|J|$ roots of the form $[0:y_2:1]$. In this case, we have
\[F(x_1,x_2,x_3)=\prod\limits_{y_3\in J}^{}(y_3x_1^b-x_3)[x_1F_1+x_2^{r_0}\prod\limits_{y_2=1}^{d_0-|J|}(x_2^b-y_2^bx_3)] .\]
But $I([0:0:1])=\langle x_1, x_2 \rangle$ and when $r_0>0$, $F'(0,0,1)=0.$ \\
When $J \neq \emptyset$, the point $[0:1:0]$ is also a root since $y_3x_1^b-x_3$ is a factor of $F$. Altogether, we have at most $2+d_0-|J|$ roots with $x_1=0$.
 Thus, we have
\begin{eqnarray}
|V_{Y}(F)| &\le& 1+d_0+1-|J|+q|J|+(q-|J|)\deg_{x_2}(F) \nonumber\\
&\leq& 1+d_0+q|J|+(q-|J|)(d-|J|b) \nonumber\\
&=& 1+d_0+qd+|J|(q-qb-d+|J|b) \nonumber\\
&\leq& 1+qd+|J|(d_0+q-qb-d+|J|b) \nonumber\\
&\leq& qd+1, \nonumber
\end{eqnarray}
since $b>1$ implies $d_0+q \leq qb$ and so we have $d_0+q-qb-d+|J|b\leq 0$. \\ Let $J=\emptyset$. Since $\mbox{deg}_{x_2}(F)\le d$, $F$ can have at most $qd$ roots of type $[1:y_2:y_3]$ by \eqref{eq:ub1}. Next, we count the number of roots with $x_1=0$. If $\deg_{x_2}(F)=d$, then $x_{2}^{d}$ would appear in $F$ so $F\notin I(0,1,0)=\langle x_1,x_3\rangle.$ Hence, $F(0,1,0)\neq 0.$ By the same reason $F$ can not be in $I(0,y_2,1)=\langle x_1, y_2^{b}x_3-x_2^{b}\rangle$ for any $y_2\in \F_{q}^{*}.$ Therefore, the only root of $F$ could be $[0:0:1]$ when $x_1=0.$ All together $F$ can have at most $qd+1$ roots. If $\mbox{deg}_{x_2}(F)\le d-1$, then by \eqref{eq:ub1}, $F$ can have at most $q(d-1)$ roots with $x_1=1$. Since $F$ can have at most $q+1$ roots with $x_1=0$, we have
\[|V_{Y}(F)|\le q(d-1)+q+1 = qd+1.\] Hence, in any case we get $|V_{Y}(F)|\le qd+1$ for $\ell=0$.

Consider now the polynomial \[F_{0}(x_1,x_2,x_3)=\prod\limits_{y_2=1}^{d}(x_2-y_2x_1)\in S_d.\] $F_0$ vanishes at $[1:y_2:y_3]$ for $d$ of $y_2\in \F_q$ and $y_3\in \F_q$ and therefore it has $qd$ roots. Also, $F_0(0,x_2,y_3)=x_2^d=0$ implies $x_2=0$, i.e. it has only one root $[0:0:1]$ when $x_1=0$. Thus, $F_0$ has exactly $qd+1$ roots. Therefore, we get finally the following equality, 
\begin{equation*}
\delta=N-(qd+1)=q^2+q+1-qd-1=q(q-d+1),\end{equation*}
as desired.

\textbf{Case V:} $\mathbf{b\le q \le d:}$ Set $d=q+kb+r$ with $0\leq r <b$. A non-zero $F\in S_d$, is of the form \begin{equation}\label{eq:pol}
F=x_1^{\ell}\prod\limits_{y_3\in J}(y_3x_1^{b}-x_3)F'(x_1,x_2,x_3)    
\end{equation} for a subset $J\subseteq \F_q$ as in Case IV and for a homogeneous polynomial $F'$ of degree $\mu=d-\ell-|J|b$. 


Let $0\le |J|\le k$. Using the facts that $f_3=x_1^{q}x_2-x_1x_2^{q}\in I(Y)$ and $|V_{Y}(F)|=|V_{Y}(\Bar{F})|$ whenever $F-\Bar{F}\in I(Y)$, we can replace $x_1x_2^{q}$ with $x_1^{q}x_2$ in $F$ and assume that $F$ has no term divisible by $x_1x_2^{q}$. If $\ell >0$, $x_1$ divides $F$ and so $\deg_{x_2}(F)<q$ yielding that the univariate polynomial $f(x_2):=F(1,x_2,y_3)\in \F_q[x_2]\setminus \{0\}$ has at most $q-1$ roots, for each $y_3\in \F_q \setminus J$. If $\ell=0$ and $f(x_2)$ has $q$ roots $y_2\in \F_q$, then $x_2-y_2$ divides $f$, meaning that $\deg_{x_2}(F)\geq q$ since $x_2-y_2$ could have been obtained from $x_2-y_2x_1$ or $y_3x_2^b-y_2^bx_3$ dividing $F$. Writing $F'=x_1F_1+F_2$ with $x_1 \nmid F_2$, we observe that $F_1=0$ as $x_2-y_2$ divides $f$ for all $y_2\in \F_q$ and $F$ has no term divisible by $x_1x_2^{q}$. So $x_2-y_2x_1$ can not divide $F'=F_2$ when $y_2\neq 0$ implying that $y_3x_2^b-y_2^bx_3$ divides $F'$ for all $y_2\in \F_q^*$. Hence, we get 
$$d-|J|b=\deg_{x_2}(F)=\deg_{x_2}(F')\geq (q-1)(q-|J|)b,$$ 
since there are $q-|J|$ such $y_3\in \F_q \setminus J$. It follows from $k\leq q-2$ and $r<b$ that 
\begin{align*}
q+(q-1)b>d=q+kb+r&\ge|J|b+(q-1)(q-|J|)b\\&=2|J|b+q^{2}b-q|J|b-qb\\&=(2|J|+q^{2}-q|J|-q)b\\&=(2-q)|J|b+q(q-1)b.
\end{align*}
Then, we have $q>(2-q)|J|b+(q(q-1)-(q-1))b=(2-q)|J|b+(q-1)^{2}b.$ Thus, $q+(q-2)|J|b>(q-1)^{2}b.$ Since $0\le |J|\le k\le q-2$, we have
\begin{align*}
 q+(q-2)^{2}b>(q-1)^{2}b\Rightarrow  q>(q-1)^{2}b-(q-2)^{2}b=(2q-3)b
\end{align*}
As $b\ge 2$, we get $q>(2q-3)b\ge 4q-6$ leading to the contradiction that $2>q$. Thus, $f$ can have at most $q-1$ roots. Therefore, if $0\le |J|\le k$, then we get the following inequality:
\begin{eqnarray*}
|V_{Y}(F)| &\le& q+1+q|J|+(q-|J|)(q-1)\\
&\le& q+1+q|J|+q^2-q-q|J|+|J|\\
&\le& q^2+k+1.
\end{eqnarray*}

On the other hand, if $|J|>k$ then we can write $|J|=k+j_{0}$ with $j_{0}\ge 1$. Let us consider the case $\ell>0$ first. In this case, $|J|<q$, for if $|J|=q$ we have $F\in I(Y)$ since the following equality holds 
\begin{equation*}\label{eq:rootj=q}
|V_{Y}(F)|=q+1+q|J|+(q-|J|)\deg_{x_2}(F)=q^{2}+q+1.
\end{equation*} 
It follows that 
\begin{align*}
d-\ell- b|J|&=d-\ell-bk-b(j_{0}-1+1)\\
&=q+r-b-b(j_{0}-1)-\ell \quad ( \text{ as } d-bk=q+r)\\
&\le q-1-\ell-b(j_{0}-1) \quad ( \text{ as } r-b\le -1) \\
&\le q-2-b(j_{0}-1) \quad ( \text{ as } -\ell \le -1).
\end{align*} Therefore we have,
\begin{align*}\label{eq:root}
|V_{Y}(F)| 
&\le q+1+q|J|+(q-|J|)(d-\ell-|J|b)\\
&\le q+1+q|J|+(q-|J|)(q-2-b(j_{0}-1)) \\
&=q+1+q|J|+q^{2}-2q-qbj_{0}+qb-q|J|+2|J|+bj_{0}|J|-b|J| \\
&=q+1+q^{2}-2q-qb(j_{0}-1)+2|J|+b|J|(j_{0}-1)\\
&=q^{2}+1-q+(j_{0}-1)b(|J|-q)+2|J|+k-k\\
&=q^{2}+k+1+(j_{0}-1)b(|J|-q)+|J|-k+|J|-q \\
&\le q^{2}+k+1+(j_{0}-1)b(-1)+j_{0}-1 \quad (\text{ as } |J|-q\le -1.)\\
&=q^{2}+k+1 -(j_{0}-1)(b+1)\\
&\le q^{2}+k+1.
\end{align*}

Let $|J|>k$ and $\ell=0.$ As $|J|=k+j_0$ where $j_0\ge 1$, we have
\begin{align*} d-bj&=d-bk-bj_0\\
&=q+r-b(j_0+1-1)\\
&=q+r-b-b(j_0-1)\\
&\le q-1-b(j_0-1)
\end{align*}
in this case which yields
\begin{align*}
|V_{Y}(F)|&\le q+1+q|J|+(q-|J|)(q-1-b(j_{0}-1))\\
&= q+1+q|J|+q^2-q-q|J|+|J|+(|J|-q)b(j_{0}-1)+k-k\\
&= q^2+k+1+(|J|-q)b(j_{0}-1)+|J|-k\\
&\le q^2+k+1-b(j_{0}-1)+j_0 \quad (\text{ if } |J|-q\le -1)\\
&\le q^2+k+1 \quad (\text{ if } j_0\ge 2 \text{ as } b\ge 2).\\
\end{align*}
Hence it remains to prove the same inequality when $j_{0}=1$ or $|J|=q$. Let $j_0=1$ so we have $|J|=k+1$. Recall that the number of roots of the form $[0:y_2:1]$ with $y_2\in \F_q^*$ is at most $d_0-|J|=d_0-k-1<q-1$ since $d_0<k+q$. Including $[0:1:0]$ and $[0:0:1]$, there are at most $q$ roots with $x_1=0.$ So, we have 
\begin{align*}
|V_{Y}(F)|&\le q+q(k+1)+(q-(k+1))(d-b(k+1))\\
&\le q+qk+q+(q-k-1)(q-1) \quad (\text{ as }q+r-b\le q-1)\\
&\le q+qk+q+q^2-q-qk+k-q+1=q^{2}+k+1.
\end{align*}
On the other hand, if $|J|=q$, then clearly we have 
\begin{align*}
|V_{Y}(F)|&\le 2+d_0-|J|+q|J|+(q-|J|)(d-|J|b)\\
&=2+d_0-q+q^2\\
&\le q^2+k+1 \quad (\text{ as } d_0<k+q)
\end{align*}
Thus, for all $\ell$ and  for all $|J|$ we have $|V_{Y}(F)| \le q^{2}+k+1.$ Consider now the following polynomial of degree $d=q+kb+r$ with $0\le r <b$,
\[F_{0}(x_1,x_2,x_3)=x_1^{r+1}\prod\limits_{y_3=1}^{k}(y_3x_1^{b}-x_3)\prod\limits_{y_2\in \F_q^{*}}(x_2-y_2x_1)\in S_d.
\] 
$F_{0}$ vanishes at $[1:y_2:y_3]$ for $k$ of $y_3$ and for all $y_2\in \F_q$ as well as it vanishes at $[1:y_2:y_3]$ for the remaining $q-k$ of $y_3$ and for all $y_2\in \F_q^*$. So it has $qk+(q-1)(q-k)=q^2-q+k$ roots. In addition, since $ r+1\ge 1$, $F_0$ vanishes at $[0:y_2:1]$ for all $y_2\in \F_q$ and at $[0:1:0].$ In total, $F_0$ has exactly $q^2+k+1$ roots.

Therefore, we get finally the following equality, 
\begin{equation*}
\delta=N-(q^2+k+1)=q^2+q+1-q^2-k-1=q-k,\end{equation*}
as desired.

\textbf{Case VI:} $\mathbf{q< b \leq d:}$
The argument used in the proof of the previous situation applies here since the proof is independent of whether $q$ is greater or less than $b$ and we get $\delta=q-k$,
as desired.

\textbf{Case VII:} $\mathbf{d=q+kb+r \mbox{ with } 0\le r<b \mbox{ and } k \ge q-1:}$

\[F_0=x_1^{\ell_{0}}\prod_{y_3\in \F_{q}^{*}}(y_3x_1^{b}-x_3)\prod_{y_2\in \F_{q}^{*}}(y_2x_1-x_2)\in S_d\] vanishes exactly at $q^2+q=q+1+(q-1)q+ (q-1)$ points, since the power $l_0:=d-q-(q-1)b+1\ge r+1\geq 1$. This means that the corresponding codeword $ev_Y(F_0)$ will have weight $1$. 
\end{proof}
\begin{remark} Conjecture 2.3 of \cite{ACGLOR2017} states an upper bound on $|V_Y(F)|$ for the set $Y=\P(1,w_1,\dots,w_n)(\F_q)$ and for any homogeneous polynomial of degree $d$ which is a multiple of $\lcm(w_1,\dots,w_n)$ proving it in the case where $m=2$ in Theorem 2.4. This is done in Theorem \ref{t:mindis} in Case $4$ with $r_0=0$.
\end{remark}
\subsection{Examples for the Main Parameters}\label{sec:expar}
In this subsection, we present tables showing the main parameters of the codes over the weighted projective spaces $X=\P(1,1,b)$. In this part our aim is to illustrate the work done throughout this article and to provide an opportunity to compare the obtained codes with the records at http://codetables.de/ (see \cite{Grassl:codetables}). The degrees shown in red in the tables show the first elements of the regularity sets. The degrees indicated in blue in the tables show that the size of a code at every b degree is equal to the length of the code.
\begin{table}[H]
\addtolength{\tabcolsep}{-3pt}
\centering
\captionof{table}{The Main Parameters of the codes on $Y=\P(1,1,b)(\F_2)$\label{Tab:mpq2}}
\vskip-1em
\begin{tabular}{ ccc }
\begin{tabular}{|l|l|l|l|l|}
\hline
$\mathbf{b}$& \textbf{Degree}& $\mathbf{N}$ & $\mathbf{K}$ & $\mathbf{\delta}$ \\ \hline
$b=2$ 
& $d=2$&7 & 4  &  2    \\ \hline
& $d=3$ & 7  &5 &  2  \\ \hline
& ${\color{darkred}d=4}$&7 & 7 &   1   \\ \hline
& $d=5$ & 7  & 6&  1  \\ \hline
& ${\color{darkblue}d=6}$& 7&  7 &    1  \\ \hline
& $d=7$ & 7 & 6&  1  \\ \hline
& ${\color{darkblue}d=8}$&7& 7&   1   \\ \hline
& $d=9$ &7& 6&  1  \\ \hline
& ${\color{darkblue}d=10}$&7&  7 &    1  \\ \hline
& $d=11$& 7&  6 &    1  \\ \hline
& ${\color{darkblue}d=12}$ & 7 & 7&  1  \\ \hline
& $d=13$&7& 6&   1   \\ \hline
& ${\color{darkblue}d=14}$ &7& 7&  1  \\ \hline
& $d=15$&7&  6 &    1  \\ \hline
\end{tabular}
\begin{tabular}{|l|l|l|l|l|}
\hline
$\mathbf{b}$& \textbf{Degree}& $\mathbf{N}$ & $\mathbf{K}$ & $\mathbf{\delta}$ \\ \hline
$b=5$  
& $d=2$&7 & 3  &  2    \\ \hline
& $d=3$ & 7  &3 &  2  \\ \hline
& $d=4$&7 & 3 & 2\\ \hline
& $d=5$ & 7  & 4&  2 \\ \hline
& $d=6$& 7&  5 &    1  \\ \hline
& $d=7$ & 7 & 6&  1  \\ \hline
& $d=8$&7& 6&   1   \\ \hline
& $d=9$ &7& 6&  1  \\ \hline
& ${\color{darkred}d=10}$&7&  7 &    1  \\ \hline
& $d=11$& 7&  6 &    1  \\ \hline
& $d=12$ & 7 & 6&  1  \\ \hline
& $d=13$&7& 6&   1   \\ \hline
& $d=14$ &7& 6&  1  \\ \hline
& ${\color{darkblue}d=15}$&7&  7 &    1  \\ \hline
\end{tabular}
\begin{tabular}{|l|l|l|l|l|}
\hline
$\mathbf{b}$& \textbf{Degree}& $\mathbf{N}$ & $\mathbf{K}$ & $\mathbf{\delta}$ \\ \hline
$b=7$ 
& $d=2$&7 & 3  &  2    \\ \hline
& $d=3$ & 7  &3 &  2  \\ \hline
& $d=4$&7 & 3 &   2  \\ \hline
& $d=5$ & 7  & 3&  2  \\ \hline
& $d=6$& 7&  3 &    2  \\ \hline
& $d=7$ & 7 & 4&  1  \\ \hline
& $d=8$&7& 5&   1   \\ \hline
& $d=9$ &7& 6&  1  \\ \hline
& $d=10$&7&  6 &    1  \\ \hline
& $d=11$& 7&  6 &    1  \\ \hline
& $d=12$ & 7 & 6&  1  \\ \hline
& $d=13$&7& 6&   1   \\ \hline
& ${\color{darkred}d=14}$ &7& 7&  1  \\ \hline
& $d=15$&7&  6 &    1  \\ \hline
\end{tabular}
\end{tabular}
\end{table}

\begin{table}[H]
\addtolength{\tabcolsep}{-3pt}
\centering
\captionof{table}{The Main Parameters of the codes on $Y=\P(1,1,b)(\F_5)$\label{Tab:mpq5}}
\vskip-1em
\begin{tabular}{ ccc }
\begin{tabular}{|l|l|l|l|l|}
\hline
$\mathbf{b}$& \textbf{Degree}& $\mathbf{N}$ & $\mathbf{K}$ & $\mathbf{\delta}$ \\ \hline
$b=2$
& $d=2$& 31& 4  & 20     \\ \hline
& $d=3$ & 31& 6& 15   \\ \hline
& $d=4$& 31 &  9 &  10    \\ \hline
& $d=5$ & 31 & 12&  5  \\ \hline
& $d=6$&31& 15  &  5    \\ \hline
& $d=7$ & 31 & 18 &  4  \\ \hline
& $d=8$& 31& 21 & 4     \\ \hline
& $d=9$ & 31 & 24&  3  \\ \hline
& $d=10$& 31& 27 &   3   \\ \hline
& $d=11$ &31 & 28 &  2  \\ \hline
& $d=12$&31 & 30 &   2   \\ \hline
& $d=13$ & 31& 30 & 1   \\ \hline
& ${\color{darkred}d=14}$& 31 & 31&   1   \\ \hline
& $d=15$ &31 & 30&  1  \\ \hline
& ${\color{darkblue}d=16}$&31 & 31 &    1  \\ \hline
& $d=17$ & 31 & 30 &  1  \\ \hline
& ${\color{darkblue}d=18}$& 31& 31  &  1    \\ \hline
& $d=19$ & 31 & 30&  1  \\ \hline
& ${\color{darkblue}d=20}$& 31& 31&  1    \\ \hline
& $d=21$ &31 & 30&   1 \\ \hline
& ${\color{darkblue}d=22}$&31 & 31&  1    \\ \hline
& $d=23$ & 31 & 30 &  1  \\ \hline
& ${\color{darkblue}d=24}$& 31& 31&   1   \\ \hline
& $d=25$ & 31 & 30&  1  \\ \hline
& ${\color{darkblue}d=26}$& 31& 31&   1   \\ \hline
& $d=27$ & 31 & 30 &  1  \\ \hline
& ${\color{darkblue}d=28}$& 31& 31&  1    \\ \hline
& $d=29$ & 31 & 30&  1  \\ \hline
& ${\color{darkblue}d=30}$& 31& 31&    1  \\ \hline
& $d=31$ & 31 & 30 &   1 \\ \hline
& ${\color{darkblue}d=32}$& 31& 31&   1   \\ \hline
& $d=33$ & 31 & 30&   1 \\ \hline
& ${\color{darkblue}d=34}$& 31& 31&   1   \\ \hline
& $d=35$ & 31 & 30&   1 \\ \hline
\end{tabular}
\begin{tabular}{|l|l|l|l|l|}
\hline
$\mathbf{b}$& \textbf{Degree}& $\mathbf{N}$ & $\mathbf{K}$ & $\mathbf{\delta}$ \\ \hline
$b=5$
& $d=2$& 31& 3 &  20    \\ \hline
& $d=3$ & 31& 4&   15 \\ \hline
& $d=4$& 31 &  5 &   10   \\ \hline
& $d=5$ & 31 & 7&   5 \\ \hline
& $d=6$&31& 8  &   5   \\ \hline
& $d=7$ & 31 & 9&  5  \\ \hline
& $d=8$& 31& 10 &   5   \\ \hline
& $d=9$ & 31 & 11&  5  \\ \hline
& $d=10$& 31& 13 &    4  \\ \hline
& $d=11$ &31 & 14 &  4  \\ \hline
& $d=12$&31 & 15 &    4  \\ \hline
& $d=13$ & 31& 16 &  4  \\ \hline
& $d=14$& 31 & 17&  4    \\ \hline
& $d=15$ &31 & 19&  3  \\ \hline
& $d=16$&31 & 20&  3    \\ \hline
& $d=17$ & 31 & 21 &  3  \\ \hline
& $d=18$& 31& 22  &  3    \\ \hline
& $d=19$ & 31 & 23&  3  \\ \hline
& $d=20$& 31& 25&  2    \\ \hline
& $d=21$ &31 & 26&  2  \\ \hline
& $d=22$&31 & 27&   2   \\ \hline
& $d=23$ & 31 & 28 &  2  \\ \hline
& $d=24$& 31& 29  &   2   \\ \hline
& ${\color{darkred}d=25}$ & 31 & 31&   1 \\ \hline
& $d=26$& 31& 30&   1   \\ \hline
& $d=27$ & 31 & 30 &  1  \\ \hline
& $d=28$& 31& 30  &  1    \\ \hline
& $d=29$ & 31 &30&   1 \\ \hline
& ${\color{darkblue}d=30}$& 31& 31&  1    \\ \hline
& $d=31$& 31& 30&  1    \\ \hline
& $d=32$ & 31 & 30 &  1  \\ \hline
& $d=33$& 31& 30  &  1    \\ \hline
& $d=34$ & 31 &30&  1  \\ \hline
& ${\color{darkblue}d=35}$& 31& 31& 1     \\ \hline
\end{tabular}
\begin{tabular}{|l|l|l|l|l|}
\hline
$\mathbf{b}$& \textbf{Degree}& $\mathbf{N}$ & $\mathbf{K}$ & $\mathbf{\delta}$ \\ \hline
$b=7$
& $d=2$& 31& 3  &   20   \\ \hline
& $d=3$ & 31& 4& 15   \\ \hline
& $d=4$& 31 &  5 & 10     \\ \hline
& $d=5$ & 31 & 6&  5  \\ \hline
& $d=6$&31& 6&  5    \\ \hline
& $d=7$ & 31 & 7 &  5  \\ \hline
& $d=8$& 31& 8 &  5    \\ \hline
& $d=9$ & 31 & 9&  5  \\ \hline
& $d=10$& 31& 10 &  5   \\ \hline
& $d=11$ &31 & 11 & 5   \\ \hline
& $d=12$&31 & 12&   4   \\ \hline
& $d=13$ & 31& 12 & 4   \\ \hline
& $d=14$& 31 & 13&   4   \\ \hline
& $d=15$ &31 & 14&  4  \\ \hline
& $d=16$&31 & 15& 4     \\ \hline
& $d=17$ & 31 & 16 & 4   \\ \hline
& $d=18$& 31& 17 &  4    \\ \hline
& $d=19$ & 31 & 18&  3  \\ \hline
& $d=20$& 31& 18&  3    \\ \hline
& $d=21$ &31 & 19&  3  \\ \hline
& $d=22$&31 & 20&   3   \\ \hline
& $d=23$ & 31 & 21 &  3  \\ \hline
& $d=24$& 31& 22&   3   \\ \hline
& $d=25$ & 31 & 23&   3 \\ \hline
& $d=26$& 31& 24&  2    \\ \hline
& $d=27$&31 & 24&   2   \\ \hline
& $d=28$ & 31 & 25 &   2 \\ \hline
& $d=29$& 31& 26&   2   \\ \hline
& $d=30$ & 31 & 27&  2  \\ \hline
& $d=31$& 31& 28&  2    \\ \hline
& $d=32$&31 & 29 &  2    \\ \hline
& $d=33$ & 31 & 30 & 1   \\ \hline
& $d=34$& 31& 30  &   1   \\ \hline
& ${\color{darkred}d=35}$ & 31 & 31&  1  \\ \hline
\end{tabular}
\end{tabular}
\end{table}

\section*{Acknowledgements} The authors thank Jade Nardi for her useful comments on the manuscript. The authors are supported by the Scientific and Technological Research Council of Turkey (T\"{U}B\.{I}TAK) under Project No: \mbox{119F177}.
The first author is supported by T\"{U}B\.{I}TAK-\mbox{2211}/A. This article is part of the first author's Ph.D. thesis under the supervision of the second author.

\end{document}